\documentclass[11pt]{article}
\usepackage{amssymb,amsmath,fullpage,graphicx}
\usepackage{color}

\newtheorem{theorem}{Theorem}
\newtheorem{lemma}{Lemma}
\newtheorem{remark}{Remark}

\def\R{\mathbb{R}}

\newenvironment{proof}{
    \noindent {\it Proof.}}{\hfill$\Box$
}
\newenvironment{proof1}{
    \noindent {\it Proof }}{\hfill$\Box$
}

\begin{document}

\begin{center}
{\Large \textbf{Global existence of small-norm solutions \\
in the reduced Ostrovsky equation}}

\vspace{0.5cm}

\textbf{Roger Grimshaw}

Department of Mathematical Sciences, Loughborough University, \\
Loughborough, LE11 3TU, UK \\

\vspace{0.2cm}

\textbf{Dmitry Pelinovsky }

Department of Mathematics, McMaster University, \\
Hamilton, Ontario, L8S 4K1, Canada

\vspace{1pt}
\end{center}

\textbf{Abstract}: We use a novel transformation of the reduced
Ostrovsky equation to the integrable Tzitz\'eica equation
and prove global existence of small-norm solutions
in Sobolev space $H^3(\mathbb{R})$. This scenario
is an alternative to finite-time wave breaking of
large-norm solutions of the reduced Ostrovsky equation.
We also discuss a sharp sufficient condition for the finite-time
wave breaking.

\section{Introduction}
\label{1}

{\em The reduced Ostrovsky equation}
\begin{equation}
\label{redOst}
(u_t + u u_x)_x = u,
\end{equation}
is the zero high-frequency dispersion limit ($\beta \to 0$)
of the Ostrovsky equation
\begin{equation}
\label{Ostrovsky}
\left( u_t + u u_x + \beta u_{xxx} \right)_x = u.
\end{equation}
The evolution equation (\ref{Ostrovsky}) was originally derived by Ostrovsky
\cite{Ostrov} to model small-amplitude long waves in a rotating
fluid of finite depth. Local and global well-posedness
of the Ostrovsky equation (\ref{Ostrovsky})
in energy space $H^1(\mathbb{R})$ was studied in recent papers \cite{GL07,LM06,Ts,VL04}.

Corresponding rigorous results for the reduced Ostrovsky equation (\ref{redOst})
are more complicated.  Local solutions exist in Sobolev
space $H^s(\mathbb{R})$ for $s > \frac{3}{2}$ \cite{SSK10}.
But for sufficiently steep initial data $u_0 \in C^1(\mathbb{R})$, local solutions break in a
finite time \cite{boyd05,Hu,LPS} in the standard sense of finite-time wave breaking
that occurs in the inviscid Burgers equation $u_t + u u_x = 0$.

However, a proof of global existence for sufficiently small
initial data has remained an open problem up to now.
In a similar equation with a cubic nonlinear term (called {\em the short-pulse equation}),
the proof of global existence was recently developed
with the help of a bi-infinite sequence of conserved quantities
\cite{PelSak}. These global solutions for small initial data coexist with
wave breaking solutions for large initial data \cite{LPS-09}.
Global existence and scattering
of small-norm solutions to zero in the generalized short-pulse equation
with quartic and higher-order nonlinear terms follow from the results of \cite{SSK10}.

Rather different sufficient conditions on the initial data
for wave breaking were obtained recently in \cite{GH} on the basis of
asymptotic analysis and supporting numerical simulations
(similar numerical simulations can be found in \cite{boyd05}).
It was conjectured  in \cite{GH} that
initial data $u_0 \in C^2(\mathbb{R})$ with $1 - 3 u''_0(x) > 0$ for all
$x \in \mathbb{R}$ generate global solutions of the reduced
Ostrovsky equation (\ref{redOst}), whereas  a sign change of
this function on the real line inevitably
leads to wave breaking in  finite time.

This paper is devoted to the rigorous proof of the first part of this conjecture,
that is,  global solutions exist for all initial data $u_0 \in H^3(\mathbb{R})$
such that $1 - 3 u_0''(x) > 0$ for all $x \in \mathbb{R}$.
Note here that if $u_0 \in H^3(\mathbb{R})$, then $u_0 \in C^2(\mathbb{R})$,
hence the function $1 - 3 u_0''(x)$ is continuous for all $x \in \mathbb{R}$
and approaches $1$ as $x \to \pm \infty$.
The second part of the conjecture is also discussed and a weaker statement
in line with this conjecture is proven.

Integrability of the reduced Ostrovsky equation was discovered
first by Vakhnenko \cite{Vakh}. In a series of papers \cite{Vakh1,Vakh2,Vakh3},
Vakhnenko, Parkes and collaborators  found and explored a
transformation of the reduced Ostrovsky equation to
the integrable {\em Hirota--Satsuma equation} with  reversed roles of the
variables $x$ and $t$. As a particular application of the
power series expansions \cite{Satsuma}, one can generate
a hierarchy of conserved quantities for the reduced Ostrovsky equation
(\ref{redOst}). This hierarchy includes the first two conserved quantities
\begin{equation}
\label{conserv-quant}
E_0 = \int_{\mathbb{R}} u^2 dx, \quad
E_{-1} = \int_{\mathbb{R}} \left[ (\partial_x^{-1} u)^2 + \frac{1}{3} u^3 \right] dx,
\end{equation}
where the anti-derivative operator is defined by the integration of $u(x,t)$ in $x$ subject to the zero-mass
constraint $\int_{\mathbb{R}} u(x,t) dx = 0$.

Higher-order conserved quantities $E_{-1}$, $E_{-2}$, and so on
involve higher-order anti-derivatives, which are defined under additional constraints on the solution $u$.
Hence,  these conserved quantities are not related to the $H^s$-norms
for positive $s$ and play no role in the study of global well-posedness of
the reduced Ostrovsky equation (\ref{redOst}) in Sobolev space  $H^s(\mathbb{R})$ for $s > \frac{3}{2}$.
Note in passing that the global well-posedness of the regular Hirota--Satsuma equation in the energy
space $H^1(\mathbb{R})$ was considered recently in \cite{iorio}.

However, a different transformation has recently been discovered for the reduced Ostrovsky equation (\ref{redOst}). 
This transformation is useful to generate a bi-infinite sequence of conserved quantities,
which are more suitable for the proof of global existence.

The alternative formulation of the integrability scheme for the reduced Ostrovsky equation starts
with the work of Hone and Wang \cite{HW}, where the reduced Ostrovsky equation
(\ref{redOst}) was obtained as a short-wave limit of the integrable {\em Degasperis--Procesi equation}.
As a result of the asymptotic reduction, these authors obtained the following Lax operator pair
for the reduced Ostrovsky equation (\ref{redOst}) in the original space and time variables:
\begin{equation}
\label{Lax}
\left\{ \begin{array}{cc} 3 \lambda \psi_{xxx} + (1 - 3 u_{xx}) \psi = 0, \\
\psi_t + \lambda \psi_{xx} + u \psi_x - u_x \psi = 0, \end{array} \right.
\end{equation}
where $\lambda$ is a spectral parameter. Note that the function $1 - 3 u_{xx}$ arises
naturally in the third-order eigenvalue problem (\ref{Lax}) in the same way as
the function $m = u - u_{xx}$ arises in another
integrable {\em Camassa--Holm equation} to determine if the global solutions or
wave breaking will occur in the Cauchy problem \cite{Constantin1,Constantin}.

More recently, based on an earlier study of Manna \& Neveu \cite{Manna}, Kraenkel {\em et al.} \cite{KLM}
found a transformation between the reduced Ostrovsky equation
(\ref{redOst}) and the integrable {\em Bullough--Dodd equation}, which is
also widely known as the {\em Tzitz\'eica equation} after its original derivation in 1910 \cite{Tzitzeica}.
In new characteristic variables $Y$ and $T$ (see section 2),
the Tzitz\'eica equation can be written in the form,
\begin{equation}
\label{Tzitzeica}
\frac{\partial^2 V}{\partial T \partial Y} = e^{-2V} - e^V.
\end{equation}
Note that the Tzitz\'eica equation is similar to the {\em sine--Gordon equation}
in characteristic coordinates, which arises in the integrability
scheme of the short-pulse equation \cite{PelSak}. Similarly to the
sine--Gordon equation, the Tzitz\'eica equation has
a bi-infinite sequence of conserved quantities, which was
discovered in two recent and independent works \cite{AAGZ,BS}.
Among those, we only need the first two conserved quantities
\begin{equation}
\label{conserv-quant-2}
Q_1 = \int_{\mathbb{R}} \left( 2 e^{V} + e^{-2V} - 3 \right) dY, \quad
Q_2 = \int_{\mathbb{R}} \left( \frac{\partial V}{\partial Y} \right)^2 dY,
\end{equation}
which were obtained from the power series expansions \cite{AAGZ}.
The conserved quantities (\ref{conserv-quant-2}) are related to the conserved quantities of the
reduced Ostrovsky equation (\ref{redOst}) in original physical variables
\begin{equation}
\label{conserv-quant-3}
E_1 = \int_{\mathbb{R}} \left[ \left( 1 - 3 u_{xx} \right)^{1/3} - 1 \right] dx, \quad
E_2 = \int_{\mathbb{R}} \frac{(u_{xxx})^2}{(1 - 3 u_{xx})^{7/3}} dx.
\end{equation}
Note that the conserved quantities (\ref{conserv-quant-3}) also appeared in
the balance equations derived in \cite{KLM}.

In Section 2, we shall use the conserved quantities $E_0$ in
(\ref{conserv-quant}) and $Q_1$, $Q_2$ in (\ref{conserv-quant-2}),
as well as the reduction to the Tzitz\'eica equation (\ref{Tzitzeica}),
to prove our main result, which is,

\begin{theorem}
Assume $u_0 \in H^3(\mathbb{R})$ such that $1 - 3 u_0''(x) > 0$ for all $x \in \mathbb{R}$.
Then, the reduced Ostrovsky equation (\ref{redOst}) admits a unique global solution
$u \in C(\mathbb{R}_+,H^3(\mathbb{R}))$.
\label{theorem-main}
\end{theorem}

It is natural to expect that the finite-time wave
breaking occurs for any $u_0 \in H^3(\mathbb{R})$ such that
$1 - 3 u_0''(x)$ changes sign for some $x \in \mathbb{R}$.
As shown in \cite{GH} in a periodic setting,
this criterion of wave breaking is sharper than the previous
criteria of wave breaking in \cite{Hu,LPS}.
Although we are not able to give a full proof of
this sharp criterion in the present work, we shall
prove the following weaker statement in Section 3:

\begin{theorem}
Assume that $u_0 \in H^3(\mathbb{R})$ is given
and there is a finite interval $[X_-,X_+]$ and a point $X_0 \in (X_-,X_+)$
such that
\begin{equation}
\label{assumption-1}
1 - 3 u_0''(x) < 0, \quad x \in (X_-,X_+),
\end{equation}
and
\begin{equation}
\label{assumption-2}
u_0'(x) < 0, \quad x \in (X_-,X_0), \qquad u_0'(x) > 0, \quad x \in (X_0,X_+),
\end{equation}
whereas $1 - 3 u_0''(x) \geq 0$ for all $x \leq X_-$ and $x \geq X_+$.
Then, a local solution $u \in C([0,t_0),H^3(\mathbb{R}))$ of
the reduced Ostrovsky equation (\ref{redOst}) breaks in a finite time
$t_0 \in (0,\infty)$ in the sense
$$
\limsup_{t \uparrow t_0} \| u(\cdot,t) \|_{H^3(\mathbb{R})} = \infty
$$
if $u_x(x_-(t),t) < 0$ and $u_x(x_+(t),t) > 0$ hold for all $t \in [0,t_0)$
along the characteristics
$x = x_{\pm}(t)$ originating from $x_{\pm}(0) = X_{\pm}$.
\label{theorem-new}
\end{theorem}

A prototypical example of the initial data for the reduced Ostrovsky equation
on the infinite line is the first Hermite function $u_0(x) = x e^{-a x^2}$,
where $a > 0$ is a parameter. A straightforward computation of the maximum of
$u_0''(x)$ shows that $u_0''(x) < \frac{1}{3}$ for all $x \in \mathbb{R}$
if $a \in (0,a_0)$, where
$$
a_0 = \frac{e^{3-\sqrt{6}}}{108 (3 - \sqrt{6})} \approx 0.0292.
$$
In this case, Theorem \ref{theorem-main} implies global existence of solutions for such initial data.
When $a > a_0$, condition (\ref{assumption-1}) is satisfied. In addition, $u_0(x)$
has a global minimum at $x = -\frac{1}{\sqrt{2a}}$ so that
condition (\ref{assumption-2}) is satisfied for $a > a_* = \frac{e}{72} \approx 0.0378$. (Note that
$a_* > a_0$.) Theorem \ref{theorem-new} implies wave breaking in
a finite time provided that {\em additional} constraints are satisfied, that is,
$u_x(x_-(t),t) < 0$ and $u_x(x_+(t),t) > 0$ hold for all times before the wave
breaking time along the characteristics $x = x_{\pm}(t)$
originating from $x_{\pm}(0) = X_{\pm}$. Although we strongly believe that these additional constraints
as well as condition (\ref{assumption-2}) are not needed for the statement
of Theorem \ref{theorem-new}, we were not able to lift out these technical restrictions.

The initial function $u_0(x) = x e^{-a x^2}$ for $a > a_*$ is shown on Fig. \ref{fig-Scheme},
where the points $X_-$, $X_+$, and $X_0$ introduced in Theorem \ref{theorem-new}
are also shown.
\begin{figure}
\begin{center}
\includegraphics[width=0.65\textwidth]{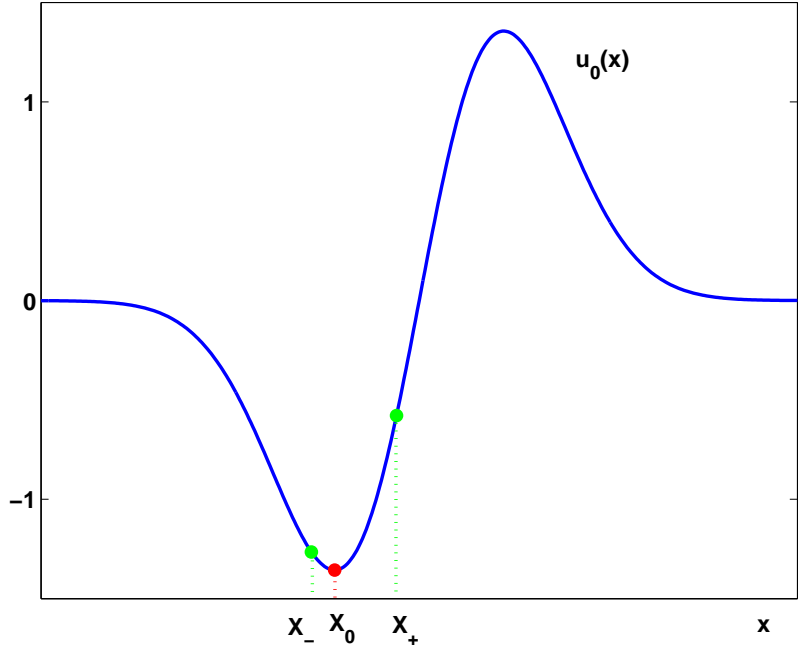}
\end{center}
\caption{An example of the initial condition $u_0(x) = x e^{-ax^2}$ with $a = 0.1$,
where positions of $X_-$, $X_+$, and $X_0$ are shown.}
\label{fig-Scheme}
\end{figure}

\section{Proof of Theorem \ref{theorem-main}}

We introduce characteristic coordinates for the
reduced Ostrovsky equation (\ref{redOst}) \cite{GH,LPS,Vakh}:
\begin{equation}
\label{characteristics}
x = X + \int_0^T U(X,T') dT', \quad t = T, \quad u(x,t) = U(X,T).
\end{equation}
The coordinate tranformation is one-to-one and onto
if the Jacobian
\begin{equation}
\label{Jacobian}
\phi(X,T) = 1 + \int_0^T U_X(X,T') dT',
\end{equation}
which is positive for $T = 0$ because $\phi(X,0) = 1$,
remains positive for all $(X,T) \in \mathbb{R} \times [0,T_0]$,
where $T_0 > 0$ is the local existence time.

We note that the original equation (\ref{redOst}) yields the relation
\begin{eqnarray}
\label{relation-1}
u & = & \frac{U_{XT}}{\phi} = \frac{\phi_{TT}}{\phi},
\end{eqnarray}
whereas the transformation formulas (\ref{characteristics}) and (\ref{Jacobian}) yield
the relations
\begin{eqnarray}
\label{relation-2}
u_x & = & \frac{U_X}{\phi} = \frac{\phi_T}{\phi},
\end{eqnarray}
and
\begin{eqnarray}
\label{relation-3}
u_{xx} & = & \frac{1}{\phi} \left( \frac{U_X}{\phi} \right)_X = \frac{\phi_{TX} \phi - \phi_X \phi_T}{\phi^3}.
\end{eqnarray}
Next, in accordance with  \cite{KLM}, we introduce the variable
\begin{equation}
\label{variable-F}
f = (1 - 3 u_{xx})^{1/3} \,.
\end{equation}
If $u$ satisfies the reduced Ostrovsky equation (\ref{redOst}),
then $f$ satisfies the balance equation
\begin{equation}
\label{balance}
f_t + (u f)_x = \frac{\{u - u_{xt} - (uu_x )_x \}_x }{f^{2/3}} = 0 \,.
\end{equation}
In characteristic coordinates (\ref{characteristics}), we set $f(x,t) = F(X,T)$,
use equation (\ref{relation-2}), and  rewrite the balance equation (\ref{balance}) in the equivalent form
\begin{equation}
\label{conservation-F}
(F \phi)_T = 0 \quad \Rightarrow  \quad F(X,T) \phi(X,T) =  F_0(X),
\end{equation}
where $F_0(X) = F(X,0)$.
Using equations (\ref{relation-3}) and (\ref{variable-F}), we obtain the evolution equation for $F(X,T)$:
\begin{equation}
\label{evolution}
\frac{\partial^2}{\partial T \partial X} \log(F)
= -\frac{\partial^2}{\partial T \partial X} \log(\phi)
= \frac{1}{3} \phi (F^3 - 1) = \frac{1}{3} F_0(X) (F^2 - F^{-1}).
\end{equation}

We shall now consider the Cauchy problem for the reduced Ostrovsky equation (\ref{redOst})
with  initial data $u_0 \in H^3(\mathbb{R})$. By the local well-posedness result \cite{SSK10},
there exists a unique local solution of the reduced Ostrovsky equation
in class $u \in C([0,t_0],H^3(\mathbb{R}))$ for some $t_0 > 0$.
By Sobolev embedding of $H^3(\mathbb{R})$ into $C^2(\mathbb{R})$,
the function $f_0(x) := (1 - 3 u_0''(x))^{1/3}$ is continuous, bounded, and
satisfies $f_0(x) \to 1$ as $|x| \to \infty$.

To prove Theorem \ref{theorem-main}, we further require that $f_0(x) > 0$
for all $x \in \mathbb{R}$, which means
from the above properties that $\inf_{x \in \mathbb{R}} f_0(x) > 0$.
Because $x = X$ for $t = T = 0$, we have $F_0 \in C(\mathbb{R})$
such that $\inf_{X \in \mathbb{R}} F_0(X) > 0$. In this case, the transformation
from $X$ to $Y$ defined by
\begin{equation}
\label{coordinates}
Y := -\frac{1}{3} \int_0^X F_0(X') d X'
\end{equation}
is one-to-one and onto for all $X \in \mathbb{R}$, because the Jacobian of the transformation is
$-\frac{1}{3} F_0(X) < 0$ and $F_0(X) \to 1$ as $|X| \to \infty$. The change of variable,
\begin{equation}
\label{change-coordinates}
F(X,T) = e^{-V(Y,T)},
\end{equation}
transforms the evolution equation (\ref{evolution}) to the
integrable Tzitz\'eica equation (\ref{Tzitzeica}).

We can now transfer the well-posedness result for local solutions of the reduced
Ostrovsky equation (\ref{redOst}) to local solutions of the Tzitz\'eica equation (\ref{Tzitzeica}).

\begin{lemma}
\label{lemma-local-existence}
Assume $u_0 \in H^3(\mathbb{R})$ such that $1 - 3 u_0''(x) > 0$ for all $x \in \mathbb{R}$.
Let
$$
V_0(Y) := -\frac{1}{3} \log(1 - 3 u_0''(x)), \quad Y := -\frac{1}{3} \int_0^x (1 - 3 u_0''(x'))^{1/3} dx'.
$$
There exists a unique local solution of the Tzitz\'eica equation (\ref{Tzitzeica})
in class $V \in C([0,T_0],H^1(\mathbb{R}))$ for some $T_0 > 0$ such that $V(Y,0) = V_0(Y)$.
\end{lemma}

\begin{proof}
We rewrite transformations (\ref{variable-F}) and (\ref{change-coordinates}) into the equivalent form,
$$
u_{xx}(x,t) = \frac{1}{3}\left(1 - f^3(x,t)\right) = \frac{1}{3} \left( 1 - e^{-3V(Y,T)} \right).
$$
The inverse of this transformation is
$$
V(Y,T) = -\frac{1}{3} \log\left(1 - 3 u_{xx}(x,t)\right).
$$
For local solutions of the reduced Ostrovsky equation (\ref{redOst}) in class
$u \in C([0,t_0],H^3(\R))$, we have
$u_{xx} \in C([0,t_0],H^1(\mathbb{R}))$ for some $t_0 > 0$ and
$u_{xx} \to 0$ as $|x| \to \infty$. We have further assumed
that $\sup_{x \in \mathbb{R}} u_0''(x) < \frac{1}{3}$, which implies that
there is $T_0 \in (0,t_0)$ such that
$\sup_{x \in \mathbb{R}} u_{xx}(x,t) < \frac{1}{3}$ for all $t \in [0,T_0]$.
Under the same condition, the transformation from $X$ to $Y$ is one-to-one and onto
for all $X \in \mathbb{R}$. Therefore, $V$ is well-defined for all $(Y,T) \in \mathbb{R} \times[0,T_0]$
and $V(Y,T) \to 0$ as $|Y| \to \infty$. By construction, $V$ is a solution of the
Tzitz\'eica equation (\ref{Tzitzeica}) and $V(Y,0) = V_0(Y)$.
It remains to show that $V$ is in class $V \in C([0,T_0],H^1(\mathbb{R}))$.

The variables $V$ and $u_{xx}$ are related by $V = u_{xx} G(u_{xx})$, where
$$
G(u_{xx}) := \frac{\log(1 - 3 u_{xx})}{(-3 u_{xx})}.
$$
Both the function $G$ and its first derivative $G'$ remain bounded in $L^{\infty}$ norm as long as
$$
\sup_{x \in \mathbb{R}} u_{xx}(x,t) < \frac{1}{3},
$$
which is satisfied for all $t \in [0,T_0]$.
Note that $G(z)$ is analytic in $z$ if $|z| < \frac{1}{3}$, but
we only need boundedness of $G(z)$ and $G'(z)$, which is achieved if $z < \frac{1}{3}$.

Next recall the transformations (\ref{characteristics}) and (\ref{coordinates}) for any function
$W(Y,T) = w(x,t)$,
\begin{eqnarray*}
\| W(\cdot,T) \|_{L^2}^2 & = & \int_{\mathbb{R}} W^2(Y,T) dY =
\frac{1}{3} \int_{\mathbb{R}} W^2(Y,T) F_0(X) dX \\
& = & \frac{1}{3} \int_{\mathbb{R}} \frac{w^2(x,t) F_0(X)}{\phi(X,T)} dx
= \frac{1}{3} \int_{\mathbb{R}} w^2(x,t) f(x,t) dx.
\end{eqnarray*}
Therefore,
$$
\| V(\cdot,T) \|_{L^2} \leq \frac{1}{\sqrt{3}} \| G(u_{xx}(\cdot,t)) \|_{L^{\infty}}
\| f(\cdot,t) \|_{L^{\infty}} \| u_{xx}(\cdot,t) \|_{L^2},
$$
which remains bounded as long as $\| u(\cdot,t) \|_{L^{\infty}}$
and $\| u_{xx}(\cdot,t) \|_{L^2}$ remain bounded.
Similarly, we can prove that $\| V_Y(\cdot,T) \|_{L^2}$ remains bounded as long as
$\| u(\cdot,t) \|_{L^{\infty}}$ and $\| u_{xxx}(\cdot,t) \|_{L^2}$ remain bounded.
Thus, we have $V \in C([0,T_0],H^1(\mathbb{R}))$ for some $T_0 > 0$.
\end{proof}

\begin{remark}
The Jacobian of the transformation from $(X,T)$ to
$(x,t)$ is given by (\ref{Jacobian}) and controlled by the relation (\ref{relation-2}).
Since $\phi(X,0) = 1$ and
\begin{equation}
\label{relation-phi-u}
\phi(X,T) = \exp \left(\int_0^T u_x(x(X,T),T) dT \right),
\end{equation}
we can see that there is $T_0 > 0$ such that $\phi(X,T) > 0$ for all $(X,T) \in \mathbb{R} \times [0,T_0]$.
Because
\begin{equation}
\label{relation-phi-F}
\phi(X,T) = \frac{F_0(X)}{F(X,T)} = F_0(X) e^{V(Y,T)},
\end{equation}
the condition $\phi(X,T) > 0$ remains true as long as $V(Y,T)$ remains bounded in $L^{\infty}$-norm.
\end{remark}

\begin{lemma}
\label{lemma-global-existence}
Let $V \in C([0,T_0],H^1(\mathbb{R}))$ for some $T_0 > 0$ be a unique local solution
of the Tzitz\'eica equation (\ref{Tzitzeica}).
Then, in fact, $V \in C(\mathbb{R}_+,H^1(\mathbb{R}))$.
\end{lemma}

\begin{proof}
We shall use $Q_1$ and $Q_2$ in (\ref{conserv-quant-2}).
The quantities are well-defined for a local solution in class $V \in C([0,T_0],H^1(\mathbb{R}))$
and conserves in time for the Tzitz\'eica equation (\ref{Tzitzeica}), according to the standard approximation arguments in
Sobolev spaces.

To be able to use $Q_1$ for the control of $\| V(\cdot,T) \|_{L^2}$,
we note that the function $H(V) := 2 e^{V} + e^{-2V} - 3$
is convex near $V = 0$ with $H(0) = H'(0) = 0$ and
$$
H''(V) = 2 e^V + 4 e^{-2V} \geq 2, \quad V \in \mathbb{R}.
$$
Therefore, $H(V) \geq V^2$ for all $V \in \mathbb{R}$, so that
$$
\| V \|_{H^1}^2 = \| V \|_{L^2}^2 + \| V_Y \|_{L^2}^2 \leq Q_1 + Q_2.
$$
By a standard continuation technique, a local solution in class $V \in C([0,T_0],H^1(\mathbb{R}))$ is uniquely
continued into a global solution in class $V \in C(\mathbb{R}_+,H^1(\mathbb{R}))$.
\end{proof}

It remains to transfer results of Lemmas \ref{lemma-local-existence} and \ref{lemma-global-existence},
as well as the $L^2$ conservation of $E_0$ in (\ref{conserv-quant}) for the proof of Theorem \ref{theorem-main}. \\

\begin{proof1}{\em of Theorem \ref{theorem-main}.}
It follows from the proof in Lemma \ref{lemma-local-existence} that
$u_{xx} = V g(V)$, where
$$
g(V) := \frac{1 - e^{-3V}}{3V}.
$$
Both the function $g$ and its first derivative $g'$
remain bounded as long as $V$ remains bounded.

By Lemma \ref{lemma-global-existence}, $V \in C(\mathbb{R}_+,H^1(\mathbb{R}))$ and hence
$F(X,T) > 0$ for all $(X,T) \in \mathbb{R} \times \mathbb{R}_+$.
Therefore, $\phi(X,T) > 0$ for all $(X,T) \in \mathbb{R} \times \mathbb{R}_+$, so that the
transformation (\ref{characteristics}) is one-to-one and onto for all $(X,T) \in \mathbb{R} \times \mathbb{R}_+$.
Using the bounded functions $g$ and $g'$, we hence have $u_{xx} \in C(\mathbb{R}_+,H^1(\mathbb{R}))$.

Finally, conservation of $E_0$ in (\ref{conserv-quant}) and the elementary Cauchy--Schwarz inequality,
$$
\| u_x \|_{L^2}^2 \leq \| u \|_{L^2} \| u_{xx} \|_{L^2},
$$
implies that  $u \in C(\mathbb{R}_+,H^3(\mathbb{R}))$.
This argument completes the proof of Theorem \ref{theorem-main}.
\end{proof1}

\section{Proof of Theorem \ref{theorem-new}}

We utilize the characteristic coordinates (\ref{characteristics})
and consider the evolution of the Jacobian $\phi$ defined by (\ref{Jacobian}). Recall
that $\phi(X,0) = 1$ whereas $F(X,0) = F_0(X) = (1 - 3 u_0''(X))^{1/3}$.
By conservation (\ref{conservation-F}), assumption (\ref{assumption-1}),
and local existence in class $u \in C([0,t_0],H^3(\mathbb{R}))$, we have
$F(X,T) < 0$ for all $X \in (X_-,X_+)$ at least for small $T \geq 0$,
whereas $F(X,T) \geq 0$ for $X \leq X_-$ and $X \geq X_+$.

Using conservation (\ref{conservation-F}) and evolution (\ref{evolution}) for $F$,
we obtain the evolution equation for $\phi(X,T)$:
\begin{equation}
\label{evolution-phi}
\frac{\partial^2}{\partial T \partial X} \log(\phi) = \frac{1}{3} \phi \left(1 - \frac{F_0^3(X)}{\phi^3}\right).
\end{equation}
Integrating this equation in $T$ with the initial condition $\phi(X,0) = 1$, we obtain
\begin{equation}
\label{monotonicity}
\frac{\partial \phi}{\partial X} = \frac{1}{3} \phi(X,T) \int_0^T \phi(X,T') \left(1 - \frac{F_0^3(X)}{\phi^3(X,T')}\right) dT'.
\end{equation}
Because the right-hand side of (\ref{monotonicity}) is positive for all $X \in (X_-,X_+)$,
the function $\phi(X,T)$ is monotonically increasing for all $X \in (X_-,X_+)$ at least for small $T \geq 0$.
Moreover, we obtain the following inequality.

\begin{lemma}
Let $\psi(X,T) := \int_0^T \phi(X,T') d T'$.
Under assumption (\ref{assumption-1}) of Theorem \ref{theorem-new}, we have
\begin{equation}
\frac{\partial \psi}{\partial X} \geq \frac{1}{6} \psi^2(X,T), \quad X \in (X_-,X_+),
\end{equation}
as long as the solution remains in class $u \in C([0,t_0],H^3(\mathbb{R}))$.
\label{lemma-1}
\end{lemma}

\begin{proof}
Because $F_0(X) < 0$ for all $X \in (X_-,X_+)$, we have from (\ref{monotonicity}):
$$
\frac{\partial \phi}{\partial X} \geq \frac{1}{3} \phi(X,T) \int_0^T \phi(X,T') dT' = \frac{1}{6} \frac{\partial}{\partial T}
\left( \int_0^T \phi(X,T') dT' \right)^2.
$$
Integrating this inequality in $T$, we obtain the assertion of the lemma.
\end{proof}

It follows from Lemma \ref{lemma-1} that
\begin{equation}
\label{estimate-below}
\frac{\partial}{\partial X} \left( -\frac{1}{\psi} \right) \geq \frac{1}{6} \quad \Rightarrow \quad
\psi(X,T) \geq \frac{6 \psi(\xi,T)}{6 - (X-\xi) \psi(\xi,T)}, \quad X \in (\xi,X_+),
\end{equation}
for any $\xi \in (X_-,X_+)$, which may depend on $T$. Therefore,
$\psi(X,T)$ becomes infinite near $X = X_+$ provided that
$(X_+ - \xi)  \psi(\xi,T) > 6$ is preserved for all $T \in [0,T_0)$, for
which the solution is defined. To ensure that this is inevitable under assumptions of Theorem \ref{theorem-new},
we prove the following result.

\begin{lemma}
Under assumptions (\ref{assumption-1}) and (\ref{assumption-2})
of Theorem \ref{theorem-new}, there exists a $C^1$
function $\xi(T)$ and $T$-independent constants
$\xi_{\pm}$ such that $\phi(\xi(T),T) = 1$, $\xi(0) = X_0 \in (X_-,X_+)$,
and $\xi(T) \in [\xi_-,\xi_+] \subset (X_-,X_+)$ for all $T \geq 0$, as long as
the solution remains in class $u \in C([0,t_0],H^3(\mathbb{R}))$
with $U_X(X_-,T) < 0$ and $U_X(X_+,T) > 0$.
\label{lemma-2}
\end{lemma}

\begin{proof}
Under assumption (\ref{assumption-2}), the function $\phi_T |_{T = 0} = U_X |_{T = 0} = u_0'(X)$
changes sign at $X = X_0$ from being negative for $X \in (X_-,X_0)$ to being
positive for $X \in (X_0,X_+)$. Therefore, we can define $\xi(0) = X_0$ and consider the level curve
$\phi(\xi(T),T) = 1$. It follows from the definition
(\ref{Jacobian}) that the function $\phi(X,T)$ is continuously differentiable in $X$ and $T$
as long as the solution remains in class $u \in C([0,t_0],H^3(\mathbb{R}))$
with
\begin{equation}
\label{ode-xi}
\frac{d \xi}{d T} = -\frac{\phi_T(\xi(T),T)}{\phi_X(\xi(T),T)} = -\frac{U_X(\xi(T),T)}{\phi_X(\xi(T),T)}.
\end{equation}
Equation (\ref{monotonicity}) implies that $\phi_X(\xi(T),T) > 0$ as long as $\xi(T)$ remains
in the interval $(X_-,X_+)$. The differential equation (\ref{ode-xi}) hence implies that if
$U_X(X_-,T) < 0$ and $U_X(X_+,T) > 0$ for all $T \geq 0$,
for which the solution is defined, then there exists $T$-independent constants $\xi_{\pm}$
such that $\xi(T) \in [\xi_-,\xi_+] \subset (X_-,X_+)$.
\end{proof}

\begin{remark}
Since $\xi(0) = X_0$ is the point of minimum of $U(X,0) = u_0(X)$ and $\phi_X(X,0) = 0$,
it follows from equation (\ref{ode-xi}) that
$$
\xi'(0) = -\frac{U_{XT}(X_0,0) + \xi'(0) U_{XX}(X_0,0)}{\phi_{XT}(X_0,T)} \quad
\Rightarrow \quad \xi'(0) = -\frac{u_0(X_0)}{2 u_0''(X_0)}.
$$
This equation shows that $\xi'(0) > 0$ if $u_0(X_0) < 0$ and $\xi'(0) < 0$ if $u_0(X_0) > 0$. Therefore,
it is not apriori clear if $\xi(T)$ can reach $X_-$ or $X_+$ in a finite time.
The restrictions $u_X(X_-,T) < 0$ and $u_X(X_+,T) > 0$ serve as a sufficient condition that
$\xi(T)$ does not reach $X_-$ and $X_+$ in a finite time, for which the solution is defined.
\end{remark}

With the help of Lemmas \ref{lemma-1} and \ref{lemma-2}, we complete
the proof of Theorem \ref{theorem-new}. \\

\begin{proof1}{\em of Theorem \ref{theorem-new}.}
We use estimate (\ref{estimate-below}) with $\xi(T)$ defined by
Lemma \ref{lemma-2}. Then, we have
\begin{equation}
\label{last-formula}
\int_0^T \phi(X,T') dT' \geq \frac{6 T}{6 - (X-\xi(T))T}, \quad X \in (\xi(T),X_+).
\end{equation}
By Lemma \ref{lemma-2}, there are $T$-independent constants $\xi_{\pm}$ such that
$\xi(T) \in [\xi_-,\xi_+] \subset (X_-,X_+)$ as long as $U_X(X_-,T) < 0$ and $U_X(X_+,T) > 0$.
The lower bound in (\ref{last-formula}) diverges at a point $X \in (\xi_+,X_+)$
if $T > \frac{6}{X_+-\xi_+}$. However,
divergence of $\int_0^T \phi(X,T') dT'$ implies divergence of $\phi(X,T)$ for some $X \in (\xi_+,X_+)$
also in a finite time $T_0 \in (0,\infty)$.
Then, equation (\ref{relation-phi-u}) shows that $u_x(x,t)$ cannot be bounded
if $\phi(X,T)$ becomes infinite for some $X \in (\xi_+,X_+)$ and some $T = T_0$,
hence the norm $\| u(\cdot,T) \|_{H^3(\mathbb{R})}$ diverges as $T \uparrow T_0$.
This argument completes the proof of Theorem \ref{theorem-new}.
\end{proof1}

\begin{remark}
Based on the asymptotic analysis and numerical simulations of \cite{GH},
we anticipate that divergence of $\phi(X,T)$ near $X = X_+$ is related to the
vanishing of $\phi(X,T)$ near $X = X_-$, such that equation (\ref{relation-2}) with
$U_X(X_-,T) < 0$ would imply that $u_x$ diverges in a finite time near $x = x_-(t)$.
However, the best that  can be obtained from equation (\ref{monotonicity}) is
\begin{equation}
\label{upper-bound-phi}
\phi(X_-,T) \leq \phi(X,T) e^{-\alpha(X)T}, \quad \alpha(X) := \frac{1}{2^{2/3}}
\int_{X_-}^X |F_0(X')| dX', \quad X \in [X_-,X_+].
\end{equation}
This upper bound is obtained from the minimization of the integrand in (\ref{monotonicity})
as follows:
$$
\phi + \frac{|F_0(X)|^3}{\phi^2} \geq \frac{3}{2^{2/3}} |F_0(X)|.
$$
If $X = \xi(T) \in (X_-,X_+)$ with $\phi(\xi(T),T) = 1$, the bound (\ref{upper-bound-phi})
only gives an exponential
decay of $\phi(X_-,T)$ to zero as $T \to \infty$. The same difficulty appears in
our attempts to use bound (\ref{upper-bound-phi}) in estimate (\ref{estimate-below}).
\end{remark}

{\bf Acknowledgement:} D.P. appreciates support and hospitality of the Department of
Mathematical Sciences of Loughborough University. The research was supported by
the LMS Visiting Scheme Program.

\end{document}